\newif\ifFull
\Fulltrue
\ifFull
\documentclass[11pt]{llncs}
\usepackage[margin=1in]{geometry}
\newcommand{\figurescale}{0.5}
\else
\documentclass{llncs}
\newcommand{\figurescale}{0.45}
\fi
\usepackage{amsmath,amssymb}
\usepackage{graphicx,hyperref,cite}
\usepackage{microtype}

\let\doendproof\endproof
\renewcommand\endproof{~\hfill\qed\doendproof}

\newcommand{\qcount}{\operatorname{\mathsf{count}}}
\newcommand{\qpart}{\operatorname{\mathsf{partition}}}
\newcommand{\starsub}{\operatorname{\mathsf{star}}}
\newcommand{\multsub}{\operatorname{\mathsf{multiply}}}
\newcommand{\inject}{\triangleleft}
\newcommand{\assign}{:=}
\newcommand{\threshold}{\tau}

\newtheorem{subroutine}{Subroutine}
\let\oldsubroutine\subroutine
\renewcommand{\subroutine}{\oldsubroutine\normalfont}

\title{From Discrepancy to Majority}
\author{David Eppstein and Daniel S. Hirschberg}
\institute{Department of Computer Science, University of California, Irvine\thanks{David Eppstein was supported in part by NSF grant  CCF-1228639.}}

\begin{document}
\maketitle

\begin{abstract}
We show how to select an item with the majority color from $n$ two-colored items, given access to the items only through an oracle that returns the discrepancy of subsets of $k$ items. We use $n/\lfloor\tfrac{k}{2}\rfloor+O(k)$ queries, improving a previous method by De Marco and Kranakis that used $n-k+k^2/2$ queries. We also prove a lower bound of $n/(k-1)-O(n^{1/3})$ on the number of queries needed, improving a lower bound of $\lfloor n/k\rfloor$ by De Marco and Kranakis.
\end{abstract}

\section{Introduction}

A large body of theoretical computer science research concerns problems of computing a function using a minimal number of calls to an oracle for another function on small subsets of input values. Such problems include sorting with a minimum number of comparisons, as well as \emph{combinatorial group testing}, in which the goal is to identify the positions of a small set of true values among a larger number of false values using an oracle that returns the disjunction of an arbitrary subset of values~\cite{DuHwa-CGT-00,EppGooHir-SJC-07}. Other problems with this flavor include Valiant's work on computing the majority of $n$ values by shallow circuits of 3-input majority gates~\cite{Val-JA-84} and recent work by the authors  using two-input disjunctions to identify a small number of slackers among a larger number of workers~\cite{EppGooHir-WADS-13}.

De Marco and Kranakis~\cite{DeMKra-DMAA-15} provide another interesting  class of such problems. Their input consists of $n$ items, each having one of two colors. The goal is to select an item of the majority color or, if the input is equally balanced between colors, to report that fact rather than returning an item. The algorithm may only access the input by \emph{counting} queries on $k$-item subsets of the input. If a subset has $b$ black items and $w=k-b$ white items, then the result of the query is $c=\min(b,w)$, the size of the smaller of the two color classes. Equivalently, one may ask for the \emph{discrepancy} $d=|b-w|$ of the query subset; the count can be calculated from the discrepancy or vice versa via the identity $2c+d=k$. The motivating application of De Marco and Kranakis is in fault diagnosis of distributed systems, which requires a majority of processors to be non-faulty. Their queries model tests that examine a small number of processors per test in order to determine whether the fault-free processors are indeed a majority.

The case $k=2$ of this problem had been previously studied~\cite{AloReiSch-IPL-93,AloReiSch-SJC-97,SakWer-Comb-91}, and optimal bounds are known~\cite{AloReiSch-IPL-93,SakWer-Comb-91}.
De Marco and Kranakis~\cite{DeMKra-DMAA-15} provide more general solutions that apply whenever $k$ is sufficiently smaller than~$n$. They show that it is possible to find a majority item for even $k$ using only $n-k+k^2/2$ counting queries,\footnote{There is a bug in their method for odd $k$, in Case 1 of Theorem 4.1, when $i=\lfloor k/2\rfloor$.} and they prove a lower bound of $\lfloor n/k\rfloor$ on the number of queries that are necessary for this problem for all~$k$.

The upper bound of De Marco and Kranakis for counting queries is greater than the lower bound by a factor of~$k$ in its leading term. In this work, we reduce this upper bound by a factor of approximately $k/2$ to $n/\lfloor\tfrac{k}{2}\rfloor+O(k)$, matching the lower bound to within a constant factor independent of~$k$. 

De Marco and Kranakis also considered a more powerful type of query, the \emph{output model}, in which the answer to a query is a partition of the queried set into two monochromatic subsets (not revealing the colors of each subset). For this problem De Marco and Kranakis provided an upper bound of $\lceil (n-1)/(k-1)\rceil$ queries, and showed that the same $\lfloor n/k\rfloor$ lower bound for counting queries applies also to the output model. For odd $k$, we show that their upper bound is tight by proving a matching lower bound. For even $k$, we slightly improve their upper bound and prove a new lower bound that is within an additive $O(n^{1/3})$ lower-order term of the upper bound. Our new lower bounds apply both to counting queries and to the output model, and the $n/(k-1)$ leading terms of the new lower bounds improve the $n/k$ leading term of the previously known bound.

Our results can also be interpreted in the framework of \emph{discrepancy theory}, the study of how small the discrepancy of the sets in a set system can be made by choosing an appropriate 2-coloring of the set elements~\cite{BecChe-08}. The first stage in our counting-query algorithm, finding an unbalanced query, is equivalent to constructing a system of $k$-element sets with discrepancy $>1$, and our results for this stage provide examples of such unbalanced $k$-set systems.

\subsection{Notational conventions and problem statement}

We use the following shorthand notation for sets:
\begin{itemize}
\item $[m]$ denotes the set $\{1,2,\dots, m\}$ of the first $m$ positive integers.
\item If $S$ is a set, $i$ is an element of $S$, and $j$ is not an element of $S$,
then $S_i^j$ denotes the set $(S\setminus\{i\})\cup\{j\}$. That is, we replace $i$ by $j$ in~$S$.
\item With the same conventions, if $A$ is a subset of $S$ and $B$ is disjoint from $S$,
then $S_A^B$ denotes the set $(S\setminus A)\cup B$.
\item If $S$ and $T$ are two sets of numbers with $|S|\ge |T|$, then $S\inject T$ is the set formed from $S$ by removing the $|T\setminus S|$ largest elements of $S\setminus T$ and replacing them by the elements of $T\setminus S$. The result is a set with the same size as $S$ that forms a subset of $S\cup T$ and a superset of $T$.
By abuse of notation, when $t$ is a number, we write $S\inject t$ as a shorthand for $S\inject\{t\}$.
\end{itemize}
To avoid confusion with the equality predicate, we use the notation $x\assign y$ to indicate that a variable $x$ of our algorithm should be assigned the new value $y$.

An instance of the majority problem may be parameterized by two values, $n$ (the number of input items) and $k$ (the size of queries), with $n>k$.
We may represent an input to the problem by an $n$-tuple $X$ of numbers $x_i$ ($i\in[n]$) where each $x_i$ is a member of the set $\{0,1\}$. The argument to a query made by a majority-finding algorithm may be represented by a set $Q\subset [n]$ with $|Q|=k$. Then we may define the results of the input queries $\qcount$ and $\qpart$ as

\begin{align*}
\qcount(Q) &= \min\left\{ \sum_{i\in Q} x_i, \sum_{i\in Q} (1-x_i)\right\}\\
\qpart(Q) &= \bigl\{ \left\{ i \mid x_i = 0 \right\}, \left\{ i \mid x_i = 1 \right\} \bigr\}.
\end{align*}
By extension, we allow these functions to be applied to any set, not necessarily of cardinality~$k$, with the same definitions.

For odd $k$ it will be convenient to partition the set $[n]$ into two complementary subsets, $M$ and~$L$. $M$ is the set of indices $i$ whose associated values $x_i$ equal the majority value of $[k]$. (This may differ from the majority of $[n]$.) Similarly, $L$ is the set of indices $i$ whose associated values $x_i$ equal the minority value in $[k]$.

We say that a query set $Q$ is \emph{homogeneous} if all of its elements have the same value; that is, it is homogeneous when $\qcount(Q)=0$ and when $\qpart(Q)=\{\emptyset,Q\}$. We say that a query is \emph{inhomogeneous} if it is not homogeneous. We say that a query set is \emph{balanced} if it is equally partitioned between elements of the two values (or as near to equal as possible when $k$ is odd). That is, $Q$ is balanced when its discrepancy is at most $1$ or when $\qcount(Q)=\lfloor k/2\rfloor$. We say that $Q$ is \emph{unbalanced} when it is not balanced.

\subsection{New results}

We prove the following new results.
\begin{itemize}
\item A majority element may be found by making $n/\lfloor\tfrac{k}{2}\rfloor+O(k)$ $\qcount$ queries. The best previous bound, by De Marco and Kranakis~\cite{DeMKra-DMAA-15}, was $n-k+k^2/2$.
\item When $n$ is odd, a majority element may be found by making $\lceil (n-2)/(k-1)\rceil$ $\qpart$ queries. This improves for some values of $k$ the best previous upper bound, by De Marco and Kranakis~\cite{DeMKra-DMAA-15}, of $\lceil (n-1)/(k-1)\rceil$.
\item Determining the majority element requires at least $\lceil (n-1)/(k-1)\rceil$ queries, for odd $k$, and $n/(k-1)-O(n^{1/3})$ queries, for even~$k$, regardless of whether the queries are of type $\qcount$ or $\qpart$. The best previous lower bound for both these query types, by De Marco and Kranakis~\cite{DeMKra-DMAA-15}, was $\lfloor n/k\rfloor$.
\end{itemize}
In addition our methods prove the following discrepancy-theoretic result:
\begin{itemize}
\item For even $k$, there exists a family of at most $2\log_2 k+1$ sets, each having $k$ elements, that cannot be $2$-colored to make every set balanced. For odd $k$, there exists a family of at most $k+3\log_2 k+4$ sets with the same property.
\end{itemize}

\section{Upper bounds for counting}

For our new upper bounds for counting we use an algorithm with the following four stages:
\begin{enumerate}
\item Find an unbalanced query $U$.
\item Use $U$ to find a homogeneous query $H$.
\item Use $H$ to determine $\qcount([n])$.
\item Based on the value of $\qcount([n])$, find the result of the majority problem.
\end{enumerate}

We describe these four stages in the following four subsections.

\subsection{Finding an unbalanced query}
Throughout this section, when a subroutine discovers that a set $U$ is unbalanced, we will abort the subroutine and its callers, and pass $U$ on to the next stage of the algorithm. To indicate that this action is not simply returning to the subroutine's caller, we describe it using the Java-like pseudocode ``throw $U$''.

For even $k$, we do not need to find an unbalanced set, as our algorithm for finding a homogeneous set does not require it. However, the solution below serves as a warmup for the odd-$k$ case. It maintains a homogeneous subset $H$ of a balanced set $B$, repeatedly doubling $H$ until it is too large to be a subset of a balanced set. To double $H$, we query a set $B_H^Q$; if it is balanced, then $Q$ and $H$ have the same composition and the doubled set $H\cup Q$ is homogeneous.
\ifFull
\newpage
\fi
\begin{subroutine}
\label{sbr:unbal-even}
to find an unbalanced set when $k$ is even:
\begin{enumerate}
\item Set $B\assign [k]$ and $H\assign\{1\}$.
\item Repeat the following steps:
\begin{enumerate}
\item If $B$ is unbalanced, throw $B$.
\item Let $Q$ be a set disjoint from $B$ with $|Q|=|H|$.
\item If $B_H^Q$ is unbalanced, throw $B_H^Q$.
\item Set $H\assign H\cup Q$ and then set $B\assign B\inject H$.
\end{enumerate}
\end{enumerate}
\end{subroutine}

\ifFull
\begin{lemma}
\autoref{sbr:unbal-even} always throws an unbalanced set after at most $2\log_2 k + 1$ queries.
\end{lemma}

\begin{proof}
\fi
Throughout the loop, $H$ remains homogeneous, and doubles in size at each iteration. The loop terminates on or before the iteration for which $k/2<|H|\le k$,  after at most $2\log_2 k+1$ queries,
because substituting such a large homogeneous set into $B$ will always produce an unbalanced set. Thus, $|H|$ cannot grow larger than $k$ and cause $B_H^Q$ to become undefined.
\ifFull
\end{proof}
\fi
For the subroutine to work correctly, we must have $n\ge 3k/2$ to ensure that  a large enough subset $Q$ disjoint from $B$ can be chosen in step~2(b).

When $k$ is odd we use a similar idea, doubling the size of a small unbalanced seed set until it overwhelms the whole set, but the details are more complicated. In the first place, the seed set for the doubling routine in the even case is always the set $\{1\}$, found without any queries, but in the odd case we choose our seed more carefully to have the form $\{j,j'\}$ where $\{j,j'\}\subset L$.
To construct this seed, we choose $j$ and $j'$ to be arbitrary indexes disjoint from $[k]$ and then verify that they both belong to $L$ by using the following subroutine:

\begin{figure}[t]
\centering
\includegraphics[scale=\figurescale]{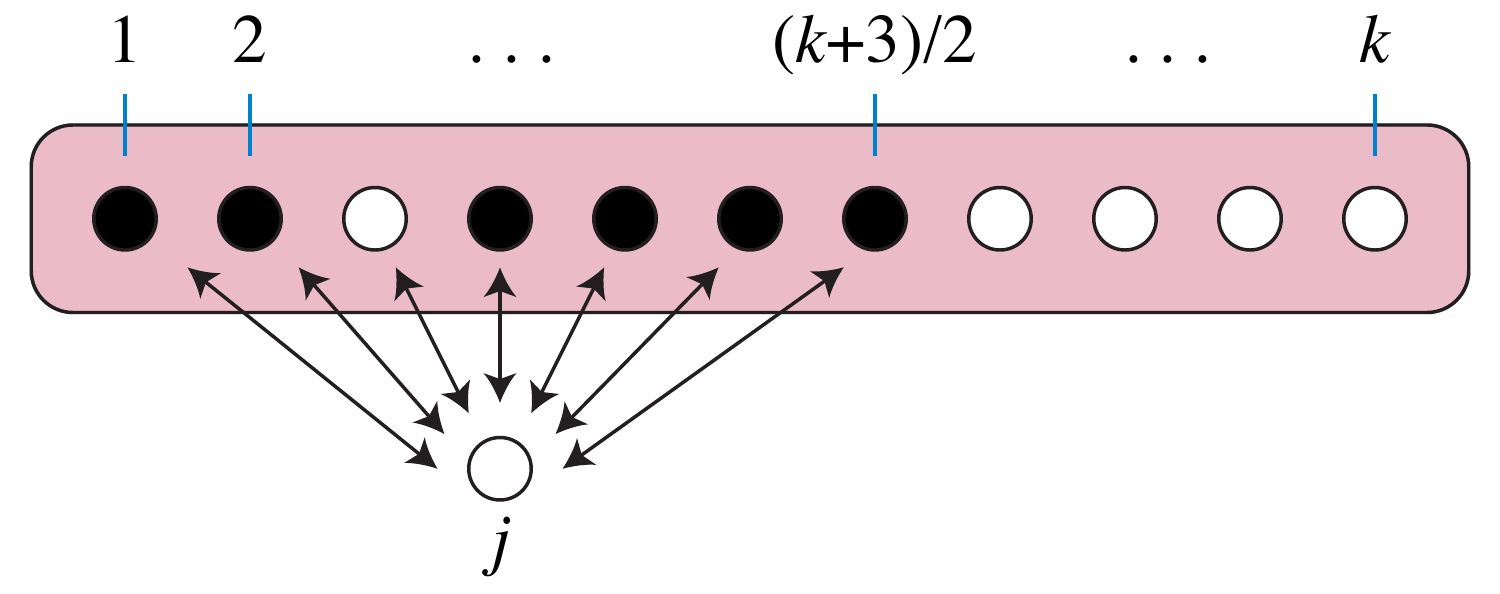}\quad
\includegraphics[scale=\figurescale]{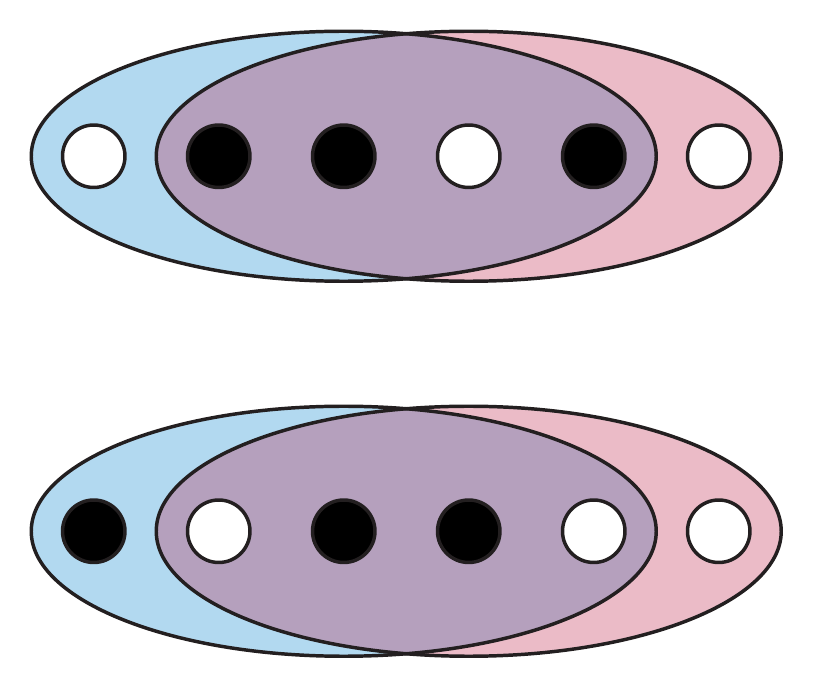}
\caption{Left: the arrows connect pairs of elements swapped into and out of the queries made by $\starsub(j)$. Right: if two overlapping queries (shown as ellipses) differ in a single element, and are both balanced, then either the two swapped elements have equal values (top) or they are unequal but both are in the majority for their query (bottom).}
\label{fig:ss-bs}
\end{figure}

\begin{subroutine} $\starsub(j)$ (for $j>k$) verifies that $j\in L$ or finds an unbalanced set:
\begin{enumerate}
\item If $[k]$ is unbalanced, throw $[k]$.
\item For $i\assign 1,2,\dots (k+3)/2$, if $[k]_i^j$ is unbalanced, throw $[k]_i^j$.
\end{enumerate}
\end{subroutine}
The subroutine name refers to the fact that the pairs $(i,j)$ defining the queries form the edges of a star graph (\autoref{fig:ss-bs}, left).

\begin{lemma}
If $\starsub(j)$ terminates without finding an unbalanced set, then $j\in L$.
\end{lemma}

\begin{proof}
There are two different possible ways that the sets $[k]$ and $[k]_i^j$ queried by the algorithm can both be balanced (\autoref{fig:ss-bs}, right): either $x_i=x_j$ (the upper case in the figure), or $i\in M$ and $j\in L$ (the lower case). The first of these two possibilities, that $x_i=x_j$, can happen only for $\lceil k/2\rceil$ choices of $i$, for otherwise too many of the members of $[k]$ would be equal to $x_j$ (and each other) for $[k]$ to be balanced. However, $\starsub(j)$ tests a larger number of pairs than that. Therefore,
if it tests all of these pairs and fails to find an unbalanced set, then it must be the case that~$j\in L$.
\end{proof}

We define a set $S$ with even cardinality to be \emph{$L$-heavy} if a majority of $S$ belongs to $L$, and \emph{$L$-balanced} if $S$ is either balanced or $L$-heavy. Because we assume $|S|$ is even, an $L$-heavy set must contain at least $1+|S|/2$ elements of $L$, and an $L$-balanced set must contain at least $|S|/2$ elements of~$L$.
\ifFull
\begin{lemma}
The disjoint union of an $L$-heavy and an $L$-balanced set is $L$-heavy.
\end{lemma}

\begin{proof}
If
\else
The disjoint union of an $L$-heavy and an $L$-balanced set must itself be $L$-heavy,
for if
\fi
$X$ and~$Y$ are disjoint with $X$ containing at least $1+|X|/2$ elements of $L$ and $|Y|$ containing at least $|Y|/2$ elements of $L$, then $X\cup Y$ contains at least $1+|X|/2+|Y|/2=1+|X\cup Y|/2$ elements of $L$.
\ifFull
\end{proof}
\fi
Our algorithm for the odd case of stage~1 depends on the following result, which lets us determine an $L$-heavy set of size double that of a previously known $L$-heavy set using $O(1)$ queries.

\begin{lemma}
\label{lem:heavy-doubling}
Suppose that $S$ and $T$ are sets disjoint from $[k]$, that $S$ is $L$-heavy, that $|S|=|T|\le k$, and that $[k]$, $[k]\inject S$, and $[k]\inject T$ are all balanced. Then $T$ is necessarily $L$-balanced.
\end{lemma}

\begin{proof}
Let $U$ be the set of the largest $|S|$ elements of $[k]$; this is the subset of $[k]$ removed to make way for $S$ in the set $[k]\inject S$ (\autoref{fig:heavy-doubling}). For $[k]$ and $[k]\inject S$ to be balanced, $U$ can have at most one more member of $M$ than $S$ does; that is, $U$ is $L$-balanced. Again, for $[k]$ and $[k]\inject T$ to be balanced, $T$ must have at least as many members of $L$ as $U$ does; therefore, $T$ is also $L$-balanced.
\end{proof}

Based on \autoref{lem:heavy-doubling}, we define a second subroutine $\multsub(P,m)$ that transforms an $L$-heavy set $P$ into a larger $L$-heavy set of cardinality $m|P|$. 
It takes as input an $L$-heavy set $P$, where $P$ has even size and is disjoint from $[k]$, and a positive integer $m$ with $m|P|\le k$. It either finds an unbalanced set~$U$ (aborting the subroutine) or returns as output an $L$-heavy set of cardinality~$m|P|$. We assume as a precondition for this subroutine that $[k]$ has already been determined to be balanced. The subroutine uses the binary representation of~$m$ to find its return value in a small number of doublings.
\begin{subroutine} $\multsub(P,m)$ (where $m$ and $P$ are as described above) finds an unbalanced set or returns an $L$-heavy set disjoint from $[k]$ of size $m|P|$:
\begin{enumerate}
\item If $m=1$, return $P$.
\item If $[k]\inject P$ is unbalanced, throw $[k]\inject P$ .
\item Choose $Q$ disjoint from both $P$ and $[k]$ with $|Q|=|P|$.
\item If $[k]\inject Q$ is unbalanced, throw $[k]\inject Q$.
\item Set $R\assign\multsub(P\cup Q,\lfloor m/2\rfloor)$.
\item If $m$ is even, return $R$.
\item Choose $S$ disjoint from $R$ and from $[k]$ with $|S|=|P|$.
\item If $[k]\inject S$ is unbalanced, throw $[k]\inject S$.
\item Return $R\cup S$.
\end{enumerate}
\end{subroutine}

\begin{figure}[t]
\centering\includegraphics[scale=\figurescale]{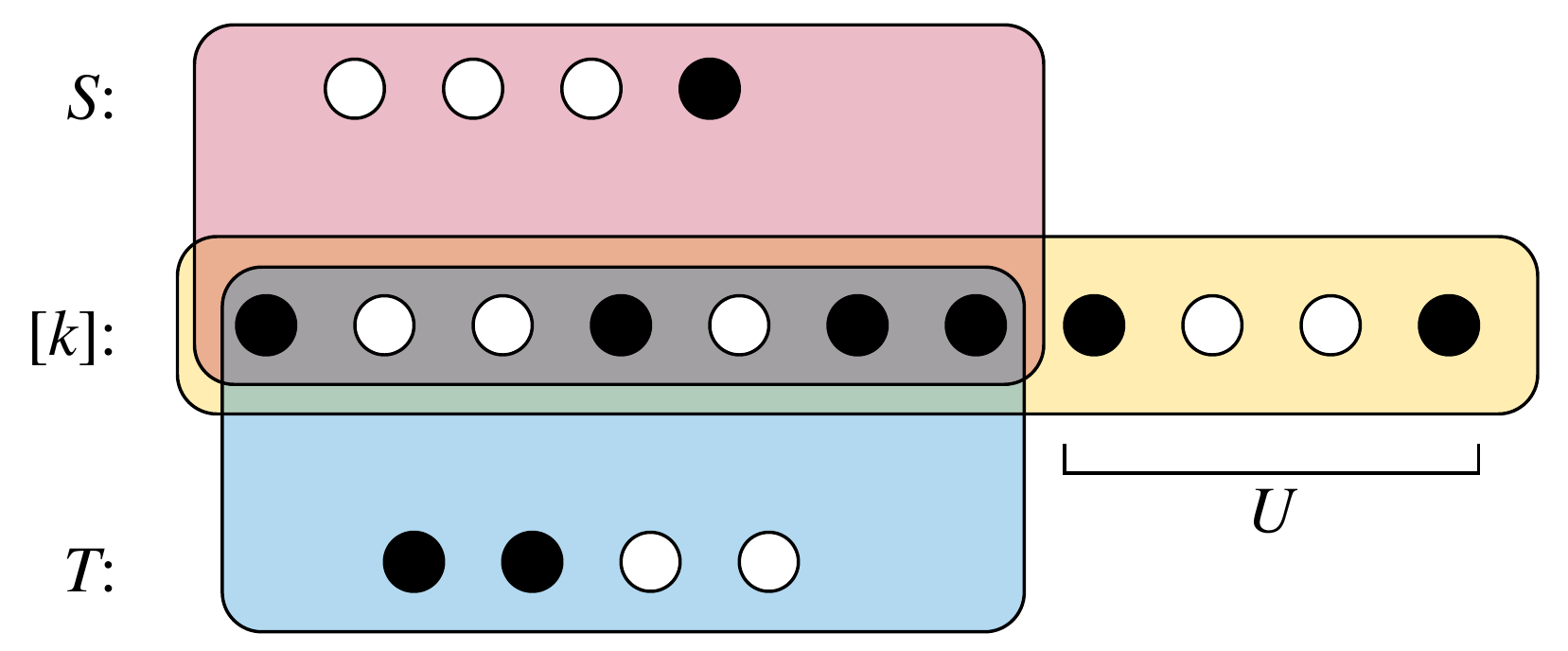}
\caption{The sets $S$ (top), $T$ (bottom), and $U$ (middle right), and the query sets $[k]$ (yellow), $[k]\inject S$ (red), and $[k]\inject T$ (blue), used in the proof of \autoref{lem:heavy-doubling}.}
\label{fig:heavy-doubling}
\end{figure}

\ifFull
\begin{lemma}
If $\multsub(P,m)$ does not throw an unbalanced set, it returns an $L$-heavy set of cardinality $m|P|$. Regardless of whether it throws an unbalanced set or returns an $L$-heavy set, it performs at most $3\log_2 m$ queries.
\end{lemma}
\begin{proof}
\fi
By \autoref{lem:heavy-doubling}, if $\multsub$ does not find an unbalanced set, then $Q$ and $S$ must both be $L$-balanced, and their disjoint union with an $L$-heavy set is another $L$-heavy set. Therefore, the set returned by this subroutine is $L$-heavy, and
(by induction on the number of recursive calls)
has the desired cardinality. The number of levels of recursion (counting only levels that can perform queries) is $\lfloor\log_2 m\rfloor$; at each level it performs either two or three queries, depending on whether $m$ is even or odd. Therefore, in the worst case, it performs at most $3\log_2 m$ queries.
\ifFull
\end{proof}
\fi

Putting $\starsub$ and $\multsub$ together, we have the following algorithm to find an unbalanced set when $k$ is odd. It uses $\starsub$ twice to find a two-element $L$-heavy set~$Y$, then uses $\multsub$ to expand this set to an $L$-heavy set of $k-1$ elements.
If this $L$-heavy set together with one element $i\in [k]$ remains unbalanced, it must be the case that $i\in M$. After we identify two members of $M$, we can replace them with the two known members of $L$ to obtain an unbalanced set.

\begin{subroutine} finds an unbalanced set when $k$ is odd:
\label{sbr:unbal-odd}
\begin{enumerate}
\item Call $\starsub(k+1)$ and $\starsub(k+2)$, and set $Y\assign\{k+1,k+2\}$.
\item Set $Z\assign\multsub(Y,(k-1)/2)$, an $L$-heavy set of $k-1$ elements.
\item If $Z\cup\{1\}$ or $Z\cup\{2\}$ is unbalanced, throw the unbalanced set.
\item Throw $[k]_{\{1,2\}}^Y$.
\end{enumerate}
\end{subroutine}

\ifFull
\begin{lemma}
\autoref{sbr:unbal-odd} always throws an unbalanced set after at most $k+3\log_2 k+3$ queries.
\end{lemma}
\begin{proof}
\fi
The two calls to $\starsub$ (after eliminating the shared query of set $[k]$) take a total of $k+4$ queries. The call to $\multsub$ takes at most $3(\log_2 k - 1)$ queries. The remaining steps of the algorithm use at most two queries. Therefore, the total number of queries made in this stage of the algorithm is at most $k+3\log_2 k+3$.
\ifFull
\end{proof}
\fi
In order to work, this algorithm needs $n$ to be at least $2k-1$ so that it can find enough elements in the disjoint sets that it chooses.

For the algorithms in this stage, the sequence of queries made by the algorithm is non-adaptive: whenever a query finds an unbalanced set, the algorithm terminates, so the sequence of queries can be found by simulating the algorithm using an oracle that knows nothing about the input and always returns a balanced result. Eventually, the algorithm will determine that some particular set is unbalanced without querying it. The sequence of query sets together with the final unqueried and unbalanced set form a family of $k$-sets with the property that, no matter how their elements are colored, at least one set in the family will be unbalanced. This proves the following result:

\begin{theorem}
When $k$ is even, there exists a family of at most $2\log_2 k + 1$ sets, each having $k$ elements, that cannot be $2$-colored to make every set in the family be balanced. When $k$ is odd, there exists a family of at most $k+3\log_2 k+4$ sets with the same property.
\end{theorem} 

These bounds are not tight for many values of~$k$. When $k=2\pmod 4$,
three $k$-sets with pairwise intersections of size $k/2$ cannot all be balanced. And for many odd values of~$k$ our bound can be improved by using optimal addition chains. However, such improvements would make our algorithms more complex and would affect only a low-order term of our overall analysis.

\subsection{Finding a homogeneous query}

\begin{figure}[t]
\centering\includegraphics[scale=\figurescale]{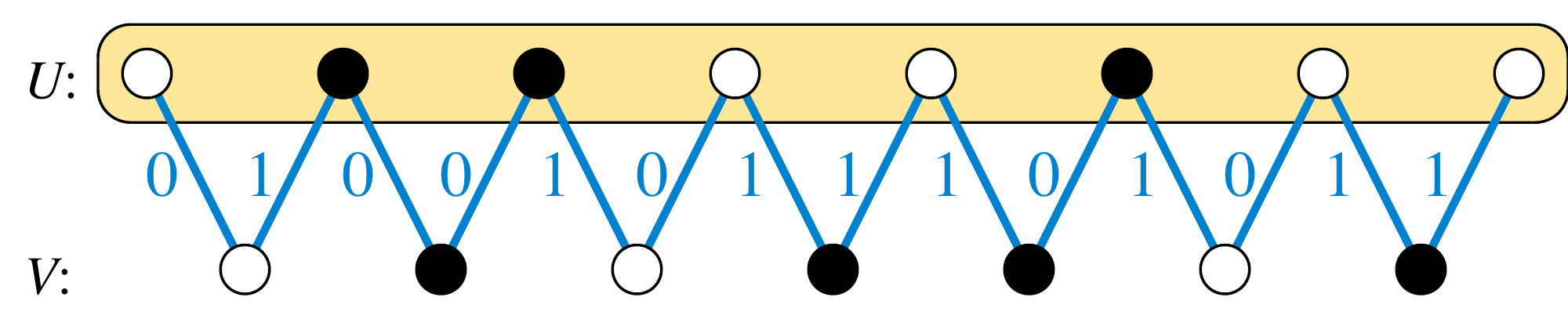}
\caption{Finding a homogeneous query. Given an unbalanced $k$-element query $U$ (top, yellow), we find a disjoint set $V$ of $k-1$ elements (bottom), and construct a spanning tree of the complete bipartite graph that has $U$ and $V$ as its two vertex sets (blue edges). We then query each set $U_i^j$ for each spanning tree edge $ij$ and use the result to label each edge $0$ (if $x_i=x_j$) or $1$ (otherwise). Any two elements of $U\cup V$ have the same value if and only if the spanning tree path connecting them has even label sum.}
\label{fig:homogenize}
\end{figure}

After the previous stage of the algorithm, we have obtained an unbalanced query~$U$. We may also assume that we know the result of the query $\qcount(U)$, for the algorithm of the previous stage will either query this number itself or it will find an unbalanced query $U$ for which $\qcount(U)$ can be determined without making a query. Our algorithm for finding a homogeneous query is based on the principle that, for any two indices $i$ and $j$ with $i\in U$ and $j\notin U$, we can test whether $x_i=x_j$ in a single additional query, by testing whether $\qcount(U_i^j)=\qcount(U)$. If $x_i=x_j$ then the count stays the same, clearly. However, with $U$ unbalanced, it is not possible for the two indices to have different values while preserving the count.

\begin{subroutine} to find a homogeneous query:
\label{sbr:homogeneous}
\begin{enumerate}
\item Let $V$ be a set of $k-1$ elements disjoint from $U$.
\item Construct a spanning tree $T$ of the complete bipartite graph $K_{k,k-1}$ having $U$ and $V$ as the two sides of its bipartition.
\item For each edge $(i,j)$ of $T$, with $i\in U$ and $j\in V$, query $\qcount(U_i^j)$. Label the edge with the number 1 if the query value is different from $\qcount(U)$ and instead label the edge with the number 0 if the two query values are equal.
\item Define two elements of $U\cup V$ to be equivalent when the path connecting them in $T$ has even label sum, and partition $U\cup V$ into the two equivalence classes $X$ and $Y$ of the resulting equivalence relation.
\item Return a subset of $k$ elements from the larger of the two equivalence classes.\end{enumerate}
\end{subroutine}

The algorithm is illustrated in \autoref{fig:homogenize}.
This stage performs $2k-2$ queries and requires that $n\ge 2k-1$.

\subsection{Finding the count}

We next use the known homogeneous query $H$ to compute $\qcount([n])$.

\begin{subroutine} to compute $\qcount([n])$, given a homogeneous set $H$:
\begin{enumerate}
\item Partition $[n]\setminus H$ into $(n-k)/\lfloor \tfrac{k}{2}\rfloor$ subsets $S_1,S_2,\dots$, each having at most $\lfloor \tfrac{k}{2}\rfloor$ elements.
\item For each subset $S_i$ of the partition, query $\qcount(H\inject S_i)$. Since $S_i\le k/2$ and the remaining elements of $H\inject S_i$ are homogeneous, this query determines the number of elements of $S_i$ that are not the same type as~$H$.
\item Let $c$ be the sum of the query values, and return $\min(c,n-c)$.
\end{enumerate}
\end{subroutine}

As well as computing $\qcount([n])$, the same algorithm can determine whether $H$ is in the majority (according to whether $c$ or $n-c$ is the minimum) and, if not, find an inhomogeneous query $I$ for which $|H\cap I|\ge k/2$ (any of the queries with a nonzero query value). The number of queries it needs is
\[
\frac{n-k}{\lfloor k/2\rfloor} \le
\frac{n}{\lfloor k/2\rfloor} -2.
\]

\subsection{Finding the majority}

After the previous three stages of the algorithm, we have the following information:
\begin{itemize}
\item A homogeneous query $H$.
\item The number $\qcount([n])$.
\item Whether the elements of $H$ are in the majority.
\item An inhomogeneous query $I$ (if $H$ is not in the majority), with $|H\cap I|\ge k/2$.
\end{itemize}

If $\qcount([n])=n/2$, we report that there is no majority. If $H$ is a subset of the majority, we may return any element of $H$ as the majority element. In the remaining case, we find an element of $I$ that is not of the same type as the elements of~$H$, using binary search:

\begin{subroutine} uses binary search to find a majority element:
\label{sbr:find-majority}
\begin{enumerate}
\item Let $U\assign I\setminus H$, a set containing an element not the same type as $H$.
\item Let $c\assign\qcount(I)$, the number of majority elements in $U$ already determined in stage three of the algorithm.
\item While $|U|>c$, do the following steps:
\begin{enumerate}
\item Let $V\assign{}$any subset of $\lfloor|U|/2\rfloor$ elements of $U$.
\item Query $\qcount(H\inject V)$.
\item If the result of the query is nonzero, let $U\assign V$ and let $c\assign{}$the query result. Otherwise, let $U\assign U\setminus V$ and leave $c$ unchanged.
\end{enumerate}
\item Return any element of the remaining set $U$. 
\end{enumerate}
\end{subroutine}

\ifFull
\begin{lemma}
\autoref{sbr:find-majority} finds a majority element using at most $\log_2 k$ additional queries.
\end{lemma}
\begin{proof}
\fi
By induction, for a given set $U$, this algorithm uses at most $\lfloor 1+\log_2(|U|-1)\rfloor$ queries. The worst case occurs when $|U|$ is one plus a power of two and the query result is zero, resulting in a case of the same type in the next step. Since initially $|U|\le k/2$, it follows that the total number of queries for this stage of the algorithm is less than $\log_2 k$.
\ifFull
\end{proof}
\fi
This bound can be improved by making a more careful choice of the set~$I$ to ensure that the initial values in the algorithm satisfy $c>|U|/2$, but this improvement is unnecessary for our results.

\subsection{Counting analysis}

By adding together the numbers of queries made in the four stages of our algorithm we obtain the following result.

\begin{theorem}
Let $k$ and $n$ be given integers with $n\ge 2k-1$ and $k>1$. Then it is possible to find a majority element of a set of $n$ $2$-colored elements, or to report that there is no majority, using at most $n/\lfloor\tfrac{k}{2}\rfloor+3k+4\log_2 k$ $\qcount$ queries on subsets of~$k$ elements.
\end{theorem}

\ifFull
In the next section
\else
In the full version of the paper
\fi
we remove the constraint that $n\ge 2k-1$ by providing substitute algorithms for the case that $k<n<2k-1$, using $O(k)$ queries.

\ifFull
\section{Upper bounds for small $n$}

Our previous algorithm for finding the majority using counting queries requires that $n\ge 2k-1$.
Here we show how to relax that assumption by giving separate algorithms for the remaining possible values of~$n$, using only $O(k)$ queries. We first consider the case that $k+1<n<2k-1$; we will handle the case that $n=k+1$ in a later subsection.
Our  algorithm for $k+1<n<2k-1$ has the following outline:
\begin{enumerate}
\item If $k$ is odd then find an unbalanced query $U$; otherwise set $U = [k]$.
\item Use $U$ to determine the result of the majority problem.
\end{enumerate}
We detail these stages in the following two subsections.

\subsection{Finding an unbalanced query when $k$ is odd}

We begin in a manner similar to what we did for the case in which $n\geq 2k-1$ and $k$ is odd.
As we already observed, the subroutine $\starsub(j)$ used in that case can fail to find an unbalanced set, but if it does then we know that $j\in L$. As in the earlier algorithm, we use the pseudocode ``throw $U$'' to indicate that the unbalanced set $U$ should be passed to the next stage of the algorithm, terminating the subroutine and any of its callers.

\begin{subroutine} to find an unbalanced query for odd $k$ and any $n\geq k+2$:
\begin{enumerate}
\item Call $\starsub(k+1)$ and $\starsub(k+2)$ to verify that $\{k+1,k+2\}\subset L$.
\item Let $Q\assign \{k+1,k+2\}$.
\item Choose $\lfloor k/2\rfloor$ disjoint pairs of elements in $[k]$. For each pair $P$,
      if $[k]_P^Q$ is unbalanced, throw $[k]_P^Q$.
\item Let $m\assign{}$the remaining unpaired element of $[k]$, and let $\{i,j\}$ be any one of the chosen pairs.
\item If $[k]_{\{i,m\}}^Q$ is unbalanced, throw $[k]_{\{i,m\}}^Q$.
\item Throw $[k]_{\{j,m\}}^Q$.
\end{enumerate}
\end{subroutine}

If any one of the sets $[k]_P^Q$ queried by the algorithm is balanced, then $P$ must contain a member of $L$. If all $\lfloor k/2\rfloor$ of these sets are balanced, then (since there are only $\lfloor k/2\rfloor$ members of $L\cap [k]$ that can be included in the pairs) each pair must consist of exactly one member of $L$ and one member of~$M$. The remaining unpaired element $m$ must also belong to $M$. Therefore, one of the two final pairs $\{i,m\}$ or $\{j,m\}$ must be a subset of $M$, and replacing it with $Q$ produces an unbalanced tuple.

This procedure works for all $n$ with $k+1<n<2k-1$ and requires at most $(3k+3)/2$ queries.

\subsection{Partitioning the input using an unbalanced query}

Our algorithm for completing the problem of finding a majority, given an unbalanced query $U$ (when $k$ is odd, or an arbitrary query $U$ when $k$ is even) is very similar to \autoref{sbr:homogeneous} for finding a homogeneous set in the case that $n$ is large.

\begin{subroutine} to find a majority element from an unbalanced query $U$:
\begin{enumerate}
\item Let $V\assign[k]\setminus U$.
\item Construct a spanning tree $T$ on the complete bipartite graph $(U,V,U\times V)$.
\item For each edge $(i,j)$ of $T$, with $i\in U$ and $j\in V$, query $\qcount(U_i^j)$ and label the edge with 0 or 1 as in \autoref{sbr:homogeneous}.
\item Partition $U\cup V$ into two subsets of vertices such that, within each subset, each pair of vertices is connected by a path in $T$ with an even label sum.
\item Return the difference in sizes of the two subsets of vertices. 
\end{enumerate}
\end{subroutine}

After finding a $U$ with a known value of $\qcount(U)$ (unbalanced in the case that $k$ is odd),
this method uses $n-1\le 2k-3$ additional queries.

\ifFull
\else

Our algorithm for completing the problem of finding a majority, given an unbalanced query $U$ (when $k$ is odd, or an arbitrary query $U$ when $k$ is even) uses a similar principle to \autoref{sbr:homogeneous} for finding a homogeneous set in the case that $n$ is large: by swapping one member of $V$ for one member of $U$ we can determine whether the two swapped values are equal or unequal.

\begin{subroutine} to find a majority element from an unbalanced query $U$:
\begin{enumerate}
\item Let $V\assign[k]\setminus U$.
\item Choose an arbitrary element $p$ of $U$.
\item Let $V_0\assign\{q\in Q\mid \qcount(U_p^q)=\qcount(U)\}$ and let $V_1\assign V\setminus V_0$.
\item If $U$ is balanced, then return an element from the larger of the two sets $V_0$ and $V_1$, or return that there is no majority if both sets have equal sizes.
\item If $V\ne V_0$, then:
\begin{enumerate}
\item Let $q$ be an arbitrary element of $V_1$.
\item If $\qcount(U_p^q)>\qcount(U)$, let $r\assign \qcount(U)+|V_1|$. If $r>n/2$, return $q$; if $r=n/2$, return that there is no majority; and if $r<n/2$, return $p$.
\item Let $r\assign \qcount(U)+|V_0|$. If $r<n/2$, return $q$; if $r=n/2$, return that there is no majority; and if $r>n/2$, return $p$.
\end{enumerate}
\item Use the fact that $V$ is homogeneous to solve the remaining problem with fewer queries.
\end{enumerate}
\end{subroutine}

\fi

\subsection{Finding a majority element when $n=k+1$}

When $n=k+1$ there are only $n$ possible queries to make: for each element, there is one query that omits that element. However, we can find a majority element using even fewer queries. The principle our algorithm uses is that, unless $[n]$ is equally split, a query that omits a minority element will have a strictly smaller $\qcount$ than a query that omits a majority element. So to find a majority element, we need only try enough queries to ensure that a majority element will be one of the omitted ones, and compare the query values.

\begin{subroutine} to find a majority element when $n=k+1$:
\begin{enumerate}
\item Let $q\assign \qcount([k])$.
\item For $i\assign 1,2,\dots q+1$, let $q_i = \qcount( [k]_i^n )$.
\item If all $q_i=q$, return that there is no majority
\item Choose $j$ with $q_j\ne q$.
\item If $q_j > q$ then return $j$.
\item If $q_j < q$ then return $n$.
\end{enumerate}
\end{subroutine}

Since $q\le k/2$ and this algorithm makes $q+2$ queries, its worst case number of queries
is $(k+4)/2$.
\fi

\section{Lower bounds}

In contrast to our upper bounds for counting queries, our lower bounds are simpler and tighter in the case that $k$ is odd, so we begin with that case first.

\ifFull
\subsection{Odd query size}
\fi

Our lower bound for odd~$k$ uses $\qpart$ queries, as they are the most powerful and can simulate  $\qcount$ queries: if it is impossible to find the majority using a given number of $\qpart$ queries, it is also impossible with the same number of queries of the other types. We prove our lower bound by an adversary argument: we design an algorithm for answering queries that, unless enough queries are made, will be able to force the querying algorithm into making a wrong choice of answer to the majority problem.

At any point during the interaction of the querying algorithm and adversary, we define the \emph{query graph} to be a bipartite graph that has the $n$ given set elements on one side of its bipartition and the queries made so far on the other side of the bipartition. We make each query be adjacent to the elements in it. As a shorthand, we use the word \emph{component} to refer to a connected component of the query graph. The querying algorithm can be assumed to know the results of applying the $\qpart$ and $\qcount$ functions to any subset of elements within a single component, as those results can be inferred from the queries actually performed within the component. Note also that, if any component $C$ has discrepancy zero, the querying algorithm may safely ignore that component for the rest of the querying process, as removing its elements from the problem will not change the majority.

To simplify the task of the adversary, we restrict the querying algorithm to make only \emph{reasonable queries}, which we define as queries that never include elements from components with zero discrepancy, and that (unless the result of the query leaves at most one nonzero-discrepancy component) never include more than one element from the same pre-query component. It follows from these properties that the querying algorithm must stop making queries, and choose an output for the majority problem, if it ever reaches a state where at most one component has nonzero discrepancy.

\begin{lemma}
Any lower bound for an algorithm that makes only reasonable queries will be valid as a lower bound for all querying algorithms.
\end{lemma}

\begin{proof}
An arbitrary querying algorithm can be transformed into one that makes only reasonable queries by skipping any query whose elements belong to one component, removing query elements that come from zero-discrepancy components or that duplicate the component of another element, and replacing the removed elements by elements from new components. This modification produces components that are supersets of the original ones, from which the results of the original queries  can be inferred.
\end{proof}

By induction, with only reasonable queries for $k$ odd, if more than one component remains, then all components have odd cardinality and therefore odd discrepancy. We design an adversary that maintains for each odd component a partition of its elements into two subsets (consistent with previous answers) that has discrepancy one. If a query produces a single component of even cardinality, we allow the adversary to choose any partition consistent with previous answers.
If a query merges multiple discrepancy-one components, then (by choosing slightly more than half of the input components to have a majority that coincides with the majority of the merged component, and slightly fewer than half of the input components to have a majority that falls into the minority of the merged component) we can always find a consistent partition with discrepancy one. Therefore, by induction, the adversary can always achieve the goals stated above.

\begin{lemma}
If a querying algorithm that makes reasonable queries does not reduce the input to a single component before producing its output, then the adversary described above can force it to compute an incorrect answer.
\end{lemma}

\begin{proof}
Unless there is one component, more than one answer to the majority problem  is consistent with the choices already made by the adversary.

In particular, if there are evenly many odd components of discrepancy one, then by choosing the majorities of all components to be the majority of the whole input, it is possible to cause the whole input to have a majority. But by choosing half of the components to have a majority of value~0 and half of the components to have a majority of value~1, it is also possible to cause the whole input to be evenly split between the two values and have no majority. Thus, regardless of whether the querying algorithm declares that there is no majority or whether it chooses a majority element, it can be made to be incorrect.

If there are an odd number of odd components, then a majority always exists. We may achieve discrepancy one for the whole input set of elements by choosing slightly more than half of the components to have majority value~1 and slightly fewer than half to have majority value~0; however, each component can be either on the majority~1 or majority~0 side, so each element can be either in the majority or in the minority. Regardless of which element the querying algorithm determines to belong to the majority, it can be made to be incorrect.
\end{proof}

\begin{theorem}
\label{thm:odd-lb}
When $k$ is odd, any algorithm that always correctly finds the majority of $n$ elements by making $\qpart$ or $\qcount$ queries must use at least $\lceil (n-1)/(k-1)\rceil$ queries.
\end{theorem}

\begin{proof}
As above, the algorithm can be assumed to make only reasonable $\qpart$ queries, and must make enough queries to reduce the query graph to a single component. This graph initially has $n$ components, and each query reduces the number of components by at most $k-1$, from which the result follows.
\end{proof}

\ifFull
\subsection{Even query size}
\fi

De Marco and Kranakis showed that the majority problem on $n$ elements may be solved using $\lceil (n-1)/(k-1)\rceil$ $\qpart$ queries on subsets of $k$ elements, matched by the lower bound of \autoref{thm:odd-lb}. For odd $n$, this bound may be improved to $\lceil (n-2)/(k-1)\rceil$ by applying it only to the first $n-1$ elements, and either returning the result (if it is a majority) or the final element (if the first $n-1$ elements have no majority). However, this modification to their algorithm can reduce the number of queries only when $k-1$ evenly divides $n-2$, which only happens when $k$ is even. Therefore, this improvement does not contradict \autoref{thm:odd-lb}.
When $k=2$ a similar improvement can be continued recursively by pairing up elements, eliminating balanced pairs, and recursively finding the majority of a set of representative elements from each pair. The resulting algorithm uses $n-b$ queries, where $b$ is the number of nonzero bits in the binary representation of $n$, and a matching lower bound is known~\cite{SakWer-Comb-91}. Again, this does not contradict \autoref{thm:odd-lb} because $k=2$ is even.
These improvements to the upper bound of De Marco and Kranakis raise the question of whether the majority can be found with significantly fewer queries whenever $k$ is even. However, we show
\ifFull
in this section
\else
in the full version of the paper
\fi
that the answer is no. An adversary strategy similar to the odd-$k$ strategy but more complicated than it can be used to prove a lower bound of $n/(k-1)-O(n^{1/3})$ on the number of queries.
\ifFull

For odd~$k$, we proved our lower bound using an adversary that (for reasonable queries) always chooses a partition of each query that gives the resulting component discrepancy exactly one. For even~$k$, we do not wish to create components of discrepancy zero (because that would allow the querying algorithm to eliminate all the elements of the component from future consideration) but, without creating components of discrepancy zero it is not possible to bound the discrepancy that may be needed. In particular, if the querying algorithm makes queries that have the structure of a complete $k$-ary tree of height~$h$, then the adversary will be forced either to create a component of discrepancy zero or to use discrepancy values as large as $2^h$. And it is not always a good strategy for the adversary to choose a nonzero discrepancy for every query, for in the case that $k=2$ and $n=2^h-1$, a complete binary tree strategy against such an adversary can succeed in answering majority queries with only $(n-1)/2$ queries by creating a homogeneous component that is large enough to overwhelm the remaining unqueried elements.

Instead, we will parameterize our adversary by a threshold value $\threshold$, and have it follow the following strategy for each query. If it is possible to partition the query elements consistently with previous queries so that the resulting component has discrepancy at most $\threshold$, choose the partition that results in as small a nonzero discrepancy as possible. Otherwise, if this is not possible, choose a partition that results in discrepancy zero. The ability to follow this strategy is ensured by the following lemmas. Here, we define the query graph and its components in exactly the same way as in the odd-$k$ lower bound, and (as in that bound) we assume that all queries are reasonable.

\begin{lemma}
Suppose that a given reasonable query combines elements from a collection of components $C_i$($i=1,2,\dots$) that have nonzero discrepancies $d_i$. Then it is possible to answer the query consistently with previous queries, to achieve total discrepancy $\sum \sigma_i d_i$ for any choice of $\sigma_i=\pm 1$.
\end{lemma}

\begin{proof}
Let $C=\bigcup C_i$ be the component resulting from the query.
If $\sigma_i=\sigma_1$ choose a partition of $C$ in which the majority elements of $C_i$ are on the same side as the majority elements of $C_1$, and if $\sigma_i\ne\sigma_1$ choose a partition of $C$ in which the majority elements of $C_i$ are on the opposite side as the majority elements of $C_1$. Then, answer the query by restricting this partition of $C$ to the query elements.
\end{proof}

\begin{lemma}
Suppose that a given reasonable query combines elements from a collection of components $C_i$ that have nonzero discrepancies $d_i$, and that at least two of the discrepancies $d_i$ are different from each other. Then it is possible to answer the query in such a way that the discrepancy of the resulting component is nonzero and at most $\max d_i$.
\end{lemma}

\begin{proof}
Again, let $C$ be the union of the query components; by the previous lemma we may choose any combination of signs for the discrepancies of these components.
We prove the lemma using induction on the number of components to be combined. We may assume without loss of generality that $d_1=\max d_i$. If some two components different from $C_1$ (without loss of generality $C_2$ and $C_3$) have different discrepancies from each other, then we may combine them with opposite signs to each other (that is, we choose $\sigma_2\ne\sigma_3$), effectively replacing them by a single component with discrepancy $|d_2-d_3|$ which is nonzero and different from~$d_1$; the result follows by induction. If on the other hand all discrepancies other than $d_1$ are smaller than $d_1$ and equal to each other, we may cancel them in pairs (again by choosing values of $\sigma_i$ with opposite signs). If this cancellation leaves no components other than $C_1$, then the discrepancy of $C$ equals $d_1$. If there is a single component~$C_2$ left after this cancellation, with discrepancy $d_2<d_1$, then by choosing opposite signs for $\sigma_1$ and~$\sigma_2$ we may cause the discrepancy of~$C$ to equal $d_1-d_2$.
\end{proof}

\begin{lemma}
Suppose that a given reasonable query combines elements from a collection of components $C_i$ that all have the same nonzero discrepancies $d$. Then if there are an odd number of components in the query, it is possible to answer the query in such a way that the discrepancy of the resulting component is exactly~$d$. If there are an even number of components, it is always possible to achieve discrepancy zero, and it is also always possible to achieve discrepancy $2d$, but it is not possible to achieve any discrepancy between those two values.
\end{lemma}

\begin{proof}
To achieve discrepancy zero or~$d$, we choose signs $\sigma_i$ that are as evenly balanced as possible between $+1$ and $-1$. To achieve discrepancy $2d$, we start with a balanced set of signs and then flip one of them.
\end{proof}

Combining these lemmas, we have the following description of the adversary's behavior.

\begin{lemma}
\label{lem:adv-strat}
For any positive even integer $\threshold$, it is possible for an adversary to answer reasonable partition queries in such a way that:
\begin{itemize}
\item Each query answer is consistent with the previous answers.
\item At all times, each nonzero discrepancy is at most $\threshold$.
\item Whenever a query results in a nonzero discrepancy $d$, then either at least one of the components combined by the query also has discrepancy at least~$d$, or all the pre-query components combined by the query have discrepancy exactly $d/2$.
\item Whenever a query results in a zero discrepancy, then all the pre-query components combined by the query have equal discrepancy that is greater than $\threshold/2$.
\end{itemize}
\end{lemma}

We have the following bound on the sizes of the zero-discrepancy components produced by this adversary:

\begin{lemma}
If an adversary follows the strategy described by \autoref{lem:adv-strat}, and a reasonable query creates a component with discrepancy $d$ that is not the single remaining component of nonzero discrepancy, then the number of elements in the component is at least $k^{\log_2 d}=d^{\log_2 k}$.
\end{lemma}

\begin{proof}
For each query $Q$ of discrepancy $d$, we can find $k$ earlier component queries of discrepancy at least $d/2$, either directly as the ones containing the elements of the component or indirectly as the set of $k$ earlier queries of the component combined by $Q$ that has discrepancy at least $d$.
By continuing recursively, we can find a complete $k$-ary tree of queries, of height $\log_2 d$, whose elements are all combined in query~$Q$. There must be at least one element of $Q$ for each leaf of this tree.
\end{proof}

\begin{corollary}
\label{cor:zero-components}
If an adversary follows the strategy described by \autoref{lem:adv-strat}, and a reasonable query creates a component with discrepancy zero, then that component must contain $\Omega(\threshold^{\log_2 k})$ elements.
\end{corollary}

To complete the lower bound argument, we must also bound the sizes of the components remaining whenever a querying algorithm has compiled enough information to correctly determine the majority.

\begin{lemma}
If an algorithm for finding the majority makes a sequence of queries that leaves a set of components with the property that each nonzero component discrepancy is at most half of the total component discrepancy, then it is impossible for the algorithm to correctly choose a majority element.
\end{lemma}

\begin{proof}
If one component $C_1$ has exactly half of the remaining discrepancy, then by setting all of the remaining component's signs $\sigma_i$ equal to each other, it is possible either to make no majority (when those signs are different from $\sigma_1$) or a majority (when those signs are equal to $\sigma_1$). Thus, regardless of whether the algorithm determines that there is or is not a majority, it may be made incorrect.

If, on the other hand, all components have discrepancy that is less than half of the total, then any element $i$ can be made to be a majority element. If $i$ belongs to any component (with zero or nonzero discrepancy), we choose equal signs for all other nonzero-discrepancy components in order to force a majority to exist, and then choose a sign for the component containing element~$i$ as desired. Thus, regardless of whether the algorithm chooses that no majority exists or chooses a particular element as a representative of the majority, it can again always be made incorrect.
\end{proof}

\begin{corollary}
\label{cor:nonzero-components}
If an adversary follows the strategy described by \autoref{lem:adv-strat}, and a querying algorithm returns its answer after making a sequence of queries that leads to components with total discrepancy at least $2\threshold$, then the adversary can choose a partition of the elements consistent with its previous answers that makes the querying algorithm incorrect.
\end{corollary}

For large values of $\threshold$, \autoref{cor:zero-components} implies that the components with zero discrepancy are also large, and therefore that there cannot be many such components.
For small values of $\threshold$, \autoref{cor:nonzero-components} implies that (when a querying algorithm is capable of determining a correct answer) the components with nonzero discrepancy have a small sum of discrepancies, and therefore that there cannot be many such components.
By choosing a value of $\threshold$ that balances the numbers of components of both types,
we can force any correct querying algorithm to leave only a small number of components of either type. To do this, it must make a large number of queries.

\begin{theorem}
When $k$ is even, any algorithm that always correctly finds the majority of $n$ elements by making $\qpart$ or $\qcount$ queries must use at least $n/(k-1) - O(n^{1/(1+\log_2 k)})$ queries.
\end{theorem}

\begin{proof}
We set $\threshold=n^{1/(1+\log_2 k)}$, and follow the adversary strategy described above.
Then by \autoref{cor:zero-components} each component with zero discrepancy left at the end of the querying algorithm must have at least $\Omega(\threshold^{\log_2 k})$ vertices, from which it follows that there are at most $n/\Omega(\threshold^{\log_2 k})=O(n^{1/(1+\log_2 k)})$ such components. By \autoref{cor:nonzero-components} the total discrepancy of the nonzero components must be less than $2\threshold$; since each such component has discrepancy at least one, the number of such components is also less than $2\threshold=O(n^{1/(1+\log_2 k)})$.

The query graph initially has $n$ components, and ends with $O(n^{1/(1+\log_2 k)})$ components.
Each query reduces the number of components by at most $k-1$, from which the result follows.
\end{proof}

In particular, for $k\ge 4$ we need at least $n/(k-1)-O(n^{1/3})$ queries. Combining this bound with the known $n-\log_2 n$ lower bound for the $k=2$ case~\cite{SakWer-Comb-91} shows that a lower bound of $n/(k-1)-O(n^{1/3})$ is valid for all~$k$.
\fi

\section{Conclusions}
We have provided new bounds for the majority problem, for $\qcount$ and $\qpart$ queries. For $\qpart$ queries with odd query size, our bounds are tight, and for even query size we achieve a matching leading term in our upper and lower bounds. However,  for $\qcount$ queries, our upper and lower bounds bounds are separated from each other by a factor of two. Reducing this gap remains open.

Recently, Gerbner et al. have given bounds for the majority problem for a different type of query that returns an element of the majority of a three-tuple~\cite{GerKesPal-EC-15}. It would be of interest to extend their results to $k$-tuples as well.

Our work also raises the discrepancy-theoretic question of how many sets are needed in a family of $k$-element sets that cannot be balanced. In this, also, our bounds are not tight and further improvement would be of interest.

{
\ifFull
\raggedright
\bibliographystyle{abuser}
\else
\bibliographystyle{splncs}
\fi
\bibliography{queries}}

\begin{thebibliography}{10}
\urlstyle{rm}

\bibitem{AloReiSch-IPL-93}
L.~Alonso, E.~M. Reingold, and R.~Schott.
\newblock {Determining the majority}.
\newblock {\em Inform. Process. Lett.} 47(5):253{--}255, 1993,
  \href{http://dx.doi.org/10.1016/0020-0190(93)90135-V}%
{doi:\nolinkurl{10.1016/0020-0190(93)90135-V}},
  \href{https://www.ams.org/mathscinet-getitem?mr=1245142}%
{MR1245142}.

\bibitem{AloReiSch-SJC-97}
L.~Alonso, E.~M. Reingold, and R.~Schott.
\newblock {The average-case complexity of determining the majority}.
\newblock {\em SIAM J. Comput.} 26(1):1{--}14, 1997,
  \href{http://dx.doi.org/10.1137/S0097539794275914}%
{doi:\nolinkurl{10.1137/S0097539794275914}},
  \href{https://www.ams.org/mathscinet-getitem?mr=1431242}%
{MR1431242}.

\bibitem{BecChe-08}
J.~Beck and W.~W.~L. Chen.
\newblock {\em {Irregularities of distribution}}.
\newblock Cambridge Tracts in Mathematics~89. Cambridge University Press,
  Cambridge, 2008, \href{https://www.ams.org/mathscinet-getitem?mr=2488272}%
{MR2488272}.

\bibitem{DeMKra-DMAA-15}
G.~De~Marco and E.~Kranakis.
\newblock {Searching for majority with $k$-tuple queries}.
\newblock {\em Discrete Math. Algorithms Appl.} 7(2):1550009, 2015,
  \href{http://dx.doi.org/10.1142/S1793830915500093}%
{doi:\nolinkurl{10.1142/S1793830915500093}},
  \href{https://www.ams.org/mathscinet-getitem?mr=3349882}%
{MR3349882}.

\bibitem{DuHwa-CGT-00}
D.-Z. Du and F.~K. Hwang.
\newblock {\em {Combinatorial Group Testing and its Applications}}.
\newblock Ser. Appl. Math.~12. World Scientific, 2nd edition, 2000,
  \href{https://www.ams.org/mathscinet-getitem?mr=1742957}%
{MR1742957}.

\bibitem{EppGooHir-SJC-07}
D.~Eppstein, M.~T. Goodrich, and D.~S. Hirschberg.
\newblock {Improved combinatorial group testing algorithms for real-world
  problem sizes}.
\newblock {\em SIAM J. Comput.} 36(5):1360{--}1375, 2007,
  \href{http://dx.doi.org/10.1137/050631847}%
{doi:\nolinkurl{10.1137/050631847}},
  \href{https://www.ams.org/mathscinet-getitem?mr=2284085}%
{MR2284085}.

\bibitem{EppGooHir-WADS-13}
D.~Eppstein, M.~T. Goodrich, and D.~S. Hirschberg.
\newblock {Combinatorial pair testing: distinguishing workers from slackers}.
\newblock {\em Proc. 13th Int. Symp. Algorithms and Data Structures (WADS
  2013)}, pp.~316{--}327. Springer, Lecture Notes in Comput. Sci. 8037, 2013,
  \href{http://dx.doi.org/10.1007/978-3-642-40104-6_28}%
{doi:\nolinkurl{10.1007/978-3-642-40104-6_28}},
  \href{https://www.ams.org/mathscinet-getitem?mr=3126368}%
{MR3126368}.

\bibitem{GerKesPal-EC-15}
D.~Gerbner, B.~Keszegh, D.~P{\'a}lv{\"o}lgyi, B.~Patk{\'o}s, M.~Vizer, and
  G.~Wiener.
\newblock {Finding a majority ball with majority answers}.
\newblock {\em Proc. 8th Eur. Conf. Combinatorics, Graph Theory, and
  Applications (EuroComb 2015)}, pp.~345{--}351. Elsevier, Elect. Notes
  Discrete Math.~49, 2015, \href{http://dx.doi.org/10.1016/j.endm.2015.06.047}%
{doi:\nolinkurl{10.1016/j.endm.2015.06.047}}.

\bibitem{SakWer-Comb-91}
M.~E. Saks and M.~Werman.
\newblock {On computing majority by comparisons}.
\newblock {\em Combinatorica} 11(4):383{--}387, 1991,
  \href{http://dx.doi.org/10.1007/BF01275672}%
{doi:\nolinkurl{10.1007/BF01275672}},
  \href{https://www.ams.org/mathscinet-getitem?mr=1137770}%
{MR1137770}.

\bibitem{Val-JA-84}
L.~G. Valiant.
\newblock {Short monotone formulae for the majority function}.
\newblock {\em J. Algorithms} 5(3):363{--}366, 1984,
  \href{http://dx.doi.org/10.1016/0196-6774(84)90016-6}%
{doi:\nolinkurl{10.1016/0196-6774(84)90016-6}},
  \href{https://www.ams.org/mathscinet-getitem?mr=756162}%
{MR756162}.

\end{thebibliography}

\end{document}